\documentclass[11pt]{article}
\usepackage{fullpage}
\usepackage{algorithm}
\usepackage{amsthm,amssymb,amsmath}
\usepackage{boxedminipage}
\usepackage[colorlinks=true,linkcolor=blue]{hyperref}
\usepackage[usenames,dvipsnames]{color}
\usepackage[american]{babel}
\usepackage[varg]{pxfonts}
\usepackage[colorlinks=true,linkcolor=blue]{hyperref}
\hypersetup{
   citecolor=blue,
   linkcolor=blue,
}
\usepackage{microtype}

\newtheorem{theorem}{Theorem}[section]

\newcommand{\R}{{\mathbb R}}

\newcommand{\C}{{\mathbb C}}
\newcommand{\Q}{{\mathbb Q}}
\newcommand{\Z}{\mathbb Z}

\newcommand{\bm}[1]{\begin{bmatrix} #1\end{bmatrix}}

\usepackage{todonotes}

\newtheorem{lemma}[theorem]{\bf Lemma}
\newtheorem{corollary}[theorem]{\bf Corollary}

\newtheorem{fact}[theorem]{\bf Fact}

\theoremstyle{definition}
\newtheorem{remark}[theorem]{\bf Remark}
\newcommand{\bit}[1]{\langle #1\rangle}

\makeatletter
\newsavebox{\@brx}
\newcommand{\llangle}[1][]{\savebox{\@brx}{\(\m@th{#1\langle}\)}%
  \mathopen{\copy\@brx\mkern2mu\kern-0.9\wd\@brx\usebox{\@brx}}}
\newcommand{\rrangle}[1][]{\savebox{\@brx}{\(\m@th{#1\rangle}\)}%
  \mathclose{\copy\@brx\mkern2mu\kern-0.9\wd\@brx\usebox{\@brx}}}
\makeatother

\newcommand{\bite}[1]{\llangle #1 \rrangle}
\newcommand{\size}[1]{\llangle #1 \rrangle}

\newcommand{\gap}{\mathrm{mingap}}
\renewcommand{\O}{O^*}

\newcommand{\ax}[1]{\widetilde{#1}}
\newcommand{\N}{\mathbb{N}}

\newcommand{\poly}{\mathrm{poly}}

\newcommand{\jnf}{\mathsf{JNF}}
\newcommand{\specfact}{\mathsf{SF}}

\newcommand{\norm}[1]{\|#1\|}
\newcommand{\maxn}[1]{\|#1\|_{\mathrm{max}}}

\newcommand{\half}{{(1/2)}}

\title{Bit Complexity of Jordan Normal Form and Polynomial Spectral Factorization}
\date{\today}
\begin{document}
\author{Papri Dey\thanks{\texttt{pdey33@gatech.edu}.} \\ Georgia Tech \and Ravi Kannan\thanks{\texttt{kannan@microsoft.com}.} \\ Microsoft Research\and Nick Ryder\thanks{\texttt{nick.ryder@berkeley.edu}.}\\ OpenAI \and Nikhil Srivastava\thanks{\texttt{nikhil@math.berkeley.edu}. Supported by NSF Grant CCF-2009011.} \\ UC Berkeley}

\maketitle

\begin{abstract}
    We study the bit complexity of two related fundamental computational problems in
linear algebra and control theory. Our results are: (1) An $\tilde{O}(n^{\omega+3}a+n^4a^2+n^\omega\log(1/\epsilon))$
time algorithm for finding an $\epsilon-$approximation to the Jordan Normal form of an integer matrix with $a-$bit entries, where $\omega$ is the exponent of matrix multiplication. (2) An
$\tilde{O}(n^6d^6a+n^4d^4a^2+n^3d^3\log(1/\epsilon))$ time algorithm for $\epsilon$-approximately computing the spectral factorization
$P(x)=Q^*(x)Q(x)$ of a given monic $n\times n$ rational matrix polynomial of
degree $2d$ with rational $a-$bit coefficients having $a-$bit common denominators, which satisfies $P(x)\succeq 0$ for all real $x$.  The first algorithm is used as a subroutine in the second one.

Despite its being of central importance, polynomial complexity bounds were not 
previously known for spectral factorization, and for Jordan form the best previous best running time was an unspecified polynomial in $n$ of degree at least twelve \cite{cai1994computing}. Our algorithms are simple and judiciously combine techniques from numerical and symbolic computation, yielding significant advantages over either approach by itself. 
\end{abstract}
\section{Introduction}
We study the bit complexity of finding approximate solutions to the following problems, where the input is assumed to be given {exactly} as a (complex) rational matrix or collection of matrices.

\begin{enumerate}
\item {\bf Jordan Normal Form.} Given $A\in \C^{n\times n}$, find a similarity $V\in\C^{n\times n}$ such
	that $A=VJV^{-1}$ where $J$ is a direct sum of {\em Jordan blocks}, i.e., matrices of type
	$$ J_\lambda:=\bm{ \lambda & 1 & 0 & 0 &\ldots \\ 0 & \lambda & 1 & 0 & \ldots\\
	0 & 0 & \lambda & 1 & \ldots\\ \vdots \\ 0 & 0 & 0 &\ldots & \lambda}$$
	for eigenvalues $\lambda\in \C$ of $A$.
	 Here $J$
		is unique up to permutations but $V$ is not if there are
		eigenspaces of dimension greater than one. The existence of the JNF is taught in
		undergraduate linear algebra courses and has myriad
		applications throughout science and mathematics.
\item {\bf Spectral Factorization.} Given an $n\times n$ monic matrix polynomial
	$$P(x)=x^{2d}I+\sum_{i\le 2d-1} x^i P_i$$ with Hermitian coefficients
		$P_i\in\C^{n\times n}$ satisfying $P(x)\succeq 0$ for all $x\in\R$, find a monic
		matrix polynomial $Q(x)=x^dI+\sum_{i\le d-1}x^iQ_i$ such that
		$P(x)=Q^*(x)Q(x)$ and $\det(Q(x))$ has all of its zeros in the
		closed upper half complex plane (where $Q^*(x)=x^dI+\sum_{i\le d-1} x^{i}Q_i^*$). Such a $Q(x)$ is guaranteed to exist and is unique \cite{rosenblatt1958multi, yakubovich1970factorization}.
		This fact has been rediscovered several times and goes under many names
		(such as matrix F\'ejer-Riesz/Wiener-Hopf factorization and matrix polynomial sum of squares) in different fields.
		 Note that the $n=1$ case is the fact that a univariate  scalar polynomial nonnegative on $\R$ may be expressed as a sum of squares (which can
		be obtained by considering the real and imaginary parts of $Q(x)$), and the $d=0$ case is just the Cholesky factorization if we allow the leading coefficient of $P(x)$ to be an arbitrary positive semidefinite matrix (not necessarily $I$).

	\end{enumerate}

Both of the above problems have generated a large literature and
several proposed methods for solving them (see Section \ref{sec:related} for a thorough
discussion).  Roughly speaking, these methods range on a spectrum between
symbolic (relying on algebraic reasoning, performing exact computations with
rational numbers, polynomials, field extensions, etc.) and numerical (relying on
analytic reasoning, semidefinite optimization, homotopy continuation, etc.). With one
exception in the case of problem (1) \cite{cai1994computing}, to the best of our knowledge none of
these methods has been rigorously shown to yield a polynomial time algorithm.

This paper provides the first polynomial time bit complexity bounds for problem (2) and
significantly improves the best known bound for (1), in the case when the input matrices have integer entries\footnote{Or are rational with a common denominator, see the corollaries following the main theorems.}. The algorithms we study
are simple and the algorithmic ingredients employed are not essentially new; rather, our main contribution is to synthesize ideas
from both the symbolic and numerical approaches to these problems, which have
in the past developed largely separately across different fields over several decades, in a way which enables good bit complexity estimates.
At a technical level, the main task is to find good bounds on both the bit lengths of
rational numbers and on the condition numbers of matrices appearing during the execution
of the algorithms. A key theme of our proofs is that bit length bounds can be used to
obtain condition number bounds and {\em vice versa}, and that carefully passing between the two
is more effective than either one alone. 

Our two main results, advertised in the abstract, appear in Sections \ref{sec:jnf} and \ref{sec:specfact} as Theorems \ref{thm:jnfmain} and \ref{thm:specfactmain}. Additional preliminaries for each result are included in its section, and further history and context for our contributions is discussed in Section \ref{sec:related}. Two notable common features of our results are:
\begin{itemize}
    \item Our algorithms have good {\em forward error} bounds, i.e., they compute approximations to the exact solution of the given instance (as opposed to backward error, computing exact solutions of nearby instances, which is the standard notion in scientific computing). This notion of error is appropriate for mathematical (as opposed to scientific) applications where discontinuous quantities in the input (such as the size of a Jordan block) can be meaningful, but typically comes at the cost of higher running times resulting from the use of numbers with large bit length.
    \item The running times of our algorithms are bounded solely in terms of the number of bits used to specify the input. This type of result is easier to use than bounds depending on difficult to compute condition numbers, especially for such ill-conditioned problems. As such, the key phenomenon enabling our results is that {\em instances of controlled bit length cannot be arbitrarily ill-conditioned} in an appropriate sense.
\end{itemize}

We conclude with a discussion and open problems in Section \ref{sec:discussion}.

\subsection{Comparison to Related Work}\label{sec:related}

\paragraph{Jordan Normal Form.} As far as we are aware, the only known polynomial bit complexity algorthm for approximately computing the JNF $A=VJV^{-1}$ of a general square rational matrix $A\in\Q^{n\times n}$ with $a-$bit entries is \cite{cai1994computing}, obtaining a runtime of $O(\poly(n,a))$ where the degree of the polynomial is not specified but is seen to be at least twelve\footnote{The related paper \cite{ar1994reliable} proposed using JNF as an ``uncheatable benchmark'' for certifying that a device has high computational power.}. 

In the symbolic computation community, the works \cite{kaltofen1986fast, ozello1987calcul, gil1992computation, giesbrecht1995nearly, roch1996fast,li1997determining} gave polynomial {\em arithmetic} complexity\footnote{The works \cite{ozello1987calcul, gil1992computation} derived bit complexity bounds for certain special cases of input matrices, but not in general.} bounds for computing the ``rational Jordan form'' of a matrix over any field. Roughly speaking, the rational Jordan form involves a symbolic representation of the matrix $J$ where the eigenvalues are represented in terms of their minimal polynomials over the field. These results are not adequate for our application to spectral factorization, which requires inverting submatrices of the similarity $V$, an operation which becomes difficult in the symbolic representation. Nonetheless our JNF algorithm is heavily inspired by the ideas in these works, relying on the same reduction to Frobenius canonical form (expressing $A$ as a direct sum of companion matrices) used in essentially all of them. The main difference is that we compute the eigenvalues approximately using numerical techniques \cite{pan2002univariate}, and are able to bound the condition number of $V$ by controlling the minimum gap between distinct eigenvalues as a function of the bit length of the input matrix.

Methods for computing the JNF must inherently involve a symbolic component since the Jordan structure can be changed by infinitesimal perturbations. It is worth mentioning that JNF is still not a solved problem ``in practice'' as trying to compute the JNF of a $50\times 50$ matrix using standard sofware packages reveals.
\paragraph{Spectral Factorization.} Polynomial spectral factorization has been rediscovered many times. The earliest references we are aware of are \cite{rosenblatt1958multi,helson1958prediction, yakubovich1970factorization,rosenblum1971factorization, choi1995sums}; the reader may consult any of the excellent surveys \cite{sayed2001survey,aylward2007explicit,dritschel2010operator,janashia2013matrix} for a detailed discussion of the history. More recently, several constructive proofs of the spectral factorization theorem have been proposed e.g. \cite{hardin2004matrix,aylward2007explicit,ephremidze2009simple,janashia2011new,ephremidze2014elementary,ephremidze2017algorithmization}, \cite[\S 2]{bakonyi2011matrix} (this list is not meant to be comprehensive). While these may be considered constructive from a mathematical standpoint, bit complexity bounds are not pursued and are not readily evident from the techniques used\footnote{This is due in each case to one or more of the following operations: solving a linear system without bounding its condition number, computing the eigenvalues of a matrix or roots of a univariate polynomial or system of multivariate polynomials ``exactly'' (which is impossible), solving a semidefinite program without controlling the volume of its feasible region, ``exactly'' computing a Schur or Jordan form of a matrix (also impossible), using an iterative scheme with no rigorous proof of convergence, and assuming arithmetic is carried out in infinite precision.}. Two particularly simple algorithms on this list are \cite{aylward2007explicit} (which requires exactly computing the Schur form of a certain matrix and inverting some of its submatrices) and \cite[\S 2]{bakonyi2011matrix} (which requires solving a semidefinite program).  We remark that unlike JNF, spectral factorization is actually a problem that is frequently solved in practice, with several of the papers above including numerical experiments.

The work most relevant for this paper is the important paper \cite{gohberg1980spectral} (see also \cite{langer1976factorization}), which  reduces spectral factorization to computing the JNF of a block companion matrix (see Section \ref{sec:specfact} for a definition), and inverting and multiplying some matrices derived from it. Our contribution is to analyze the conditioning of this approach and combine it with our JNF algorithm, yielding concrete bit complexity bounds.

One notable advantage of our algorithm is that it works even when the input is degenerate --- i.e., $P(x)$ is only positive semidefinite rather than positive definite --- which frequently occurs in applications. This is in contrast to almost all of the works mentioned above, which only consider the strictly positive definite case (or even require all roots of $\det(P(x))$ to be distinct) and appeal to nonconstructive limiting arguments to handle the degenerate case.

A more stringent variant of the problem is to find a real factorization $P(x)=Q^T(x)Q(x)$ in the case when $P(x)$ is real symmetric, possibly allowing $Q(x)$ to be rectangular. The recent works \cite{blekherman2019low,hanselka2019positive} have obtained optimal bounds on the dimensions of $Q(x)$. In this paper, we restrict our attention to the Hermitian setting.


\subsection{Preliminaries}\label{sec:prelim}
\renewcommand{\b}{\mathsf{dy}}
\newcommand{\round}{\mathsf{round}}
{\em Asmyptotic Notation.} We will use $\O(\cdot)$ to suppress polylogarithmic factors in the input parameters
$n$ (dimension), $a$ (bit length of input numbers), $d$ (degree), and  $b$ (desired bits of accuracy).  Logarithmic factors are not the focus of this paper and can be safely ignored everywhere because all proofs in this paper invoke this notation at most a constant number of times (in particular, our algorithms do not contain any loops which could lead to blowups in the exponents of the logarithms). \\

\noindent {\em Numbers and Arithmetic.} We say that $x\in \Z\bite{a}$ if $x$ is an integer with bit length at most $a$ and $x\in \Z\bit{a}$ if $x$ is an integer with bit length at most $\O(a)$.  We use   $\Q\bit{a/c}$ (resp. $\Q\bite{a/c}$) to denote the rationals $p/q$ with $p\in\Z\bit{a},q\in\Z\bit{c}$ (resp. $\Z\bite{a},\Z\bite{c}$), and $\Q_\b\bit{a/c}$ to denote the elements of $\Q\bit{a/c}$ with denominator equal to a power of two; the latter will sometimes be useful since adding rationals with dyadic denominators does not increase the bit length of the denominator. For a rational $x$, let $\round_c(x)$ denote the nearest rational with denominator $2^c$, which clearly satisfies
\begin{equation}\label{eqn:round}
|x-\round_c(x)|\le 2^{-c}\end{equation}
and can be computed in time nearly linear in the bit length of $x$.
This notation extends to complex numbers with rational real and imaginary parts in the natural way. The bit complexity of arithmetic with rational numbers is nearly linear in the bit length (see e.g. \cite{grotschel2012geometric}).\\

\newcommand{\K}{\mathbb{K}}
\noindent{\em Matrices.} We use $\Z^{n\times n}\bite{a}$ (resp. $\Z^{n\times n}\bit{a}$) to denote integer matrices with entries of bit length $a$ (resp. $\O(a)$), and $\Q^{n\times n}\bit{a/c}$ (similarly $\Q^{n\times n}_\b\bit{a/c}$, $\C^{n\times n}\bit{a/c}$, and $\C^{n\times n}_\b\bit{a/c}$) to denote matrices with entries in $\Q\bit{a/c}$ having a {\em common denominator}. For $A\in\Z^{n\times n}$, the notation $\size{A}$ refers to the maximum bit length of an entry of $A$.

We record the following easy facts about inverses as well as products and sums of pairs of matrices\footnote{We do not rely on matrix arithmetic with a superconstant number of matrices in this paper, for which the bit length bounds necessarily depend on the number of matrices.}, which follow from the adjugate formula for the inverse and the assumption on common denominators\footnote{Allowing distinct denominators in the entries of $A,B$ could increase the bit lengths of $AB$ and $A+B$ by a factor of $n$ if the denominators are, say, relatively prime}:
\begin{fact}[Bit Length of Matrix Arithmetic]\label{fact:arithmetic}
\begin{enumerate} 
    \item If $A\in\Z^{n\times n}\bit{a}$ then $A^{-1}\in \Q^{n\times n}\bit{an/an}$.
    \item If $A\in\K^{n\times n}\bit{a/c}$ then $A^{-1}\in \K^{n\times n}\bit{c+an/an}$, for $\K=\Q,\Q_\b,\C,\C_\b$.
    \item If $A,B\in\K^{n\times n}\bit{a/c}$ then
    $$A+B\in \K^{n\times n}\bit{a/c}\quad\textrm{and}\quad AB\in \K^{n\times n}\bit{a/c},$$
    for $\K=\Q,\Q_\b,\C,\C_\b$.
\end{enumerate}
            
\end{fact}

\noindent {\em Perturbation Theory.} We use $\|\cdot\|$ to denote the operator norm and $\maxn{\cdot}$ to denote the entrywise $\ell_\infty$ norm of a matrix, noting that $\maxn{M}\le \|M\|\le n\maxn{M}$ for an $n\times n$ matrix $M$.
We use $\kappa(M):=\|M\|\|M^{-1}\|$ to denote the condition number of an invertible matrix. We will frequently use the elementary fact:
\begin{equation}\label{eqn:kappainv}
	\|(M+E)^{-1}-M^{-1}\| \le \frac{\|E\|\|M^{-1}\|}{1-\|E\|\|M^{-1}\|}\cdot \|M^{-1}\|
\end{equation}
provided $\|E\|\|M^{-1}\|<1$, which follows from a Neumann series argument, as well as its consequence
\begin{equation}\label{eqn:kappakappa}
	\kappa(M+E) \le \kappa(M)\frac{1+\|E\|\|M^-1\|}{1-\|E\|\|M^{-1}\|}
\end{equation}
whenever $\|E\|\|M^{-1}\|<1$. \\

\noindent {\em Polynomials.}  We use $\gap(\cdot)$ to indicate the minimum gap between {\em distinct} roots of a polynomial.  We use $\maxn{P(\cdot)-Q(\cdot)}$ to denote the coefficient-wise $\maxn{\cdot}$ norm of two matrix polynomials. We extend the notations $\bite{\cdot},\bit{\cdot}$ to polynomials by applying them to each scalar or matrix coefficient.\\

\noindent {\em Bit Length of Inverse and Characteristic Polynomial.} We will frequently appeal to the bounds
\begin{equation}\label{eqn:invbound}\|A^{-1}\|\le n!2^{an}\quad\textrm{whenever $A\in \Z^{n\times n}\bite{a}$ is invertible},\end{equation}
which is easily seen by considering the adjugate formula for the inverse, and
\begin{equation}\label{eqn:charbound}
\size{\chi_A(x)}\le n!2^{an}\quad\textrm{for $A\in\Z^{n\times n}\bite{a}$},
\end{equation}
where $\chi_A(x):=\det(xI-A)$ denotes the characteristic polynomial.\\

\section{Jordan Normal Form}\label{sec:jnf}
The {\em companion matrix} of a scalar monic polynomial $p(x) = x^d+\sum_{i<d} p_ix^i$ is the $d\times d$ matrix:
\begin{equation}\label{eqn:scalarcomp}
	C_p:=\bm{ 	0& 1	& 	& 	&	&	\\
			0&	0&	1&	&	&	\\
			&	&	\vdots 	&	&	\\
			&	&	&	&	&1	\\
			-p_0&	-p_1&	-p_2&	-p_3& \ldots& -p_{d-1}}^T
\end{equation}
It is easily seen that $\det(xI-C_p)=p(x)$. The high level idea of our algorithm is to use symbolic techniques to reduce the input matrix to a direct sum of companion matrices, and then use explicit formulas and root finding algorithms to compute the JNF of the companion matrices. We will rely on the following tools.

\begin{theorem}[Exact Frobenius Canonical Form, \cite{giesbrecht2002computing} {Theorems 2.2 \& 3.2}]\label{thm:storjohann} There is a randomized Las Vegas algorithm which given $A\in \Z^{n\times n}\bite{a}$, outputs a matrix $F\in \Z\bit{an}$ which is a direct sum of companion matrices and an invertible $U\in\Z\bit{an^2}$ satisfying $A=UFU^{-1}$, with an expected running time of $\O(n^5a+n^4a^2)$ bit operations.
\end{theorem}

\begin{theorem}[Approximate Polynomial Roots in the Unit Disk, \cite{pan2002univariate} Corollary 2.1.2\footnote{The parameter $b'$ here corresponds to $b/n$ in \cite{pan2002univariate}}] \label{thm:pan} There is an algorithm which given bitwise access\footnote{i.e., the algorithm can query the $i$th bit of the binary expansion of each coefficient in constant time. This is slightly different from the access model in this paper where rational numbers are given as numerator and denominator, but it is easy to see that given a rational in $\Q\bite{a/c}$, the desired binary expansion needed to apply Theorem \ref{thm:pan} can be produced in $O^*(a+c)$ time.} to the coefficients of a polynomial $p\in\Q[x]$ of degree $n$ with all roots $z_1,\ldots,z_n\in\C$ satisfying $|z_i|\le 1$ and a parameter $b'\ge \log n$, computes  numbers $z_1',\ldots,z_n'\in \C_\b\bit{b'}$ such that $|z_i'-z_i|\le 2^{2-b'}$ for all $i\le n$, using at most $\O(n^2b')$ bit operations.\end{theorem}

\begin{theorem}[Minimum Gap of Integer Polynomials, \cite{mahler1964inequality}] \label{thm:mignotte} If $p\in \Z[x]\bite{a}$ is monic of degree $n$, then $$\gap(p)\ge 2^{-an-2n\lg n}.$$
\end{theorem}

\begin{corollary}[Approximate Roots and Multiplicities of Integer Polynomials]\label{cor:rootmult} There is an algorithm which given an integer polynomial $p\in \Z[x]\bite{a}$ of degree $n\ge 2$ with roots $z_1,\ldots,z_n\in\C$ and a parameter $b'\ge an+4n\lg n$, computes numbers $z_1',\ldots,z_n'\in \C_\b\bit{(a+b')/b'}$ such that $|z_i'-z_i|<2^{-b'}$ for all $i\le n$, using at most $\O(n^2(b'+a))$ bit operations. Each $z_i'$ appears a number of times exactly equal to the multiplicity of $z_i$ in $p(x)$.
\end{corollary}
\begin{proof} The largest root of $p(x)$ has magnitude at most the sum of the absolute values of its coefficients, which is at most $M=n2^a$. Apply Theorem \ref{thm:pan} to the polynomial $p(Mx)$, which has roots in the unit disk, with error parameter $b'+\lg(M)=b'+\lg(n)+a+1$ to obtain numbers $\tilde{z_1},\ldots,\tilde{z_n}\in \C_\b\langle (a+b')/(a+b')\rangle=\C_\b\langle b'/b'\rangle$ (since $b'\ge a$) with common denominator. Then for $i=1,\ldots n$ we have $z_i':=\tilde{z_i} M \in \C\langle (a+b')/b'\rangle$ with common denominator and $|z_i'-z_i|\le 2^{-b'}$. By Theorem \ref{thm:mignotte} the minimum gap
		between distinct $z_i$ is at least $2^{-b'+1}$ since $2n\lg n \ge 1$ so this is sufficient
		to correctly determine the multiplicity of each $z_i$ and replace all $z_i'$ corresponding to a root with the same value.
\end{proof}

We will also use an explicit formula for the JNF of a companion matrix as a confluent Vandermonde matrix (see e.g. \cite{gautschi1962inverses, batenkov2012norm} for a discussion) in the roots of the corresponding polynomial.
\begin{theorem}[\cite{brand1964companion}] \label{thm:brand} If $C\in\C^{n\times n}$ is a companion matrix with distinct eigenvalues $\lambda_1,\ldots \lambda_k\in\C$ of multiplicities $m_1,\ldots,m_k$, then $C=WJW^{-1}$ with
	$$J = \oplus_{j\le k} J_{\lambda_j}$$
	$$W = [W_{\lambda_1}, W_{\lambda_2},\ldots, W_{\lambda_k}]$$
	where $J_{\lambda_j}$ an $m_j\times m_j$ Jordan block with eigenvalue $\lambda_j$ and $W_{\lambda_j}$ is the $n\times m_j$ matrix:
	\begin{equation}\label{eqn:brand} W_{\lambda_j}:=\bm{	1 & 0 & 0 & 				\ldots & 0\\
				\lambda_j & 1 & 0 & 			\ldots & 0\\
				\lambda_j^2 & 2\lambda_j & 1 &		\ldots & 0\\
				\lambda_j^3 & 3\lambda_j^2 & \binom{3}{2}\lambda_j	&	\ldots & 0\\
				& & \vdots & &\\
				\lambda_j^{n-1} & (n-1)\lambda_j^{n-2} & \binom{n-1}{2}\lambda_j^{n-3} & \ldots & \binom{n-1}{m_j}\lambda_j^{n-m_j}},
		\end{equation}
	so that $W$ is a confluent Vandermonde matrix. 
\end{theorem}

Note that the entries of $W_\lambda$ in \eqref{eqn:brand} are univariate polynomials of degree $n$ in the $\lambda_i$ with coefficients in $\Z\bite{n}$. We now present the algorithm.

\begin{figure}[ht]
\begin{boxedminipage}{\textwidth}
Algorithm $\jnf$.\\
Input: $A\in \Z^{n\times n}\bite{a}$, desired bits of accuracy $b$.\\
Output: $\ax{J}, \ax{V}\in \C_\b^{n\times n}\bit{an^3+b/(an^3+b)}.$\\
Guarantee: $\|J-\ax{J}\|\le 2^{-b}\|J\|, \|V-\ax{V}\|\le 2^{-b}\|V\|$ for some exact JNF $A=VJV^{-1}$ and $\kappa(\ax{V})\le 2^{\O(an^3)}$.

\begin{enumerate}
	\item Exactly compute the Frobenius Normal Form $A=UFU^{-1}$ with
		$F\in\Z^{n\times n}\bit{an}$ and $U\in\Z^{n\times
		n}\bit{an^2}$ using Theorem \ref{thm:storjohann}. Let
		$F=\oplus_{i\le \ell} C_i$ for companion matrices
		$C_i\in\Z^{n_i\times n_i}\bit{an}$. Let 
		$ c:=\max_i \size{C_i}.$

	\item For $i=1\ldots \ell$, apply Corollary \ref{cor:rootmult} to the
		characteristic polynomial $\chi_{C_i}(x)\in \Z[x]\bit{an}$ 
		with accuracy \begin{equation}\label{eqn:bset} b':=b+(n+1)\size{U}+cn^2+an^2+4n^2\lg n+ 7an+3\lg n\end{equation} to obtain
		approximations $\ax{\lambda_{i1}},\ldots,
		\ax{\lambda_{ik_i}}\in\C_\b\bit{b'/b'}$ to the distinct
		eigenvalues $\lambda_{i1},\ldots,\lambda_{ik_i}$ of $C_i$, with
		error $|\lambda_{ij}-\ax{\lambda_{ij}}|\le 2^{-b'}$, as
		well as their multiplicities.

	\item For $i=1,\ldots,\ell$, compute  approximate eigenvalue powers $\ax{\lambda_{i1}^p},\ldots,\ax{\lambda_{ik_{i}}^p}\in\C_\b\bit{b'/b'}$ using Lemma \ref{lem:approxpowers}. Let 
		$$\ax{J_i}:=\oplus_{j\le k_i} J_{\ax{\lambda_{ij}}}\in\C_\b^{n_i\times n_i}\bit{b'/b'}\qquad \textrm{and } \qquad \ax{W_i}:=[W_{\ax{\lambda_{i1}}}, \ldots, W_{\ax{\lambda_{ik_i}}}]\in\C_\b^{n_i\times n_i}\bit{b'/b'}$$
		as in \eqref{eqn:brand}, i.e., substitute the approximate powers $\ax{\lambda_{ij}^p}$ into the appropriate polynomials $J_\lambda, W_{\lambda}$.
	\item Output $\ax{J}$ and $\ax{V}=U\ax{W}$.
\end{enumerate}
\end{boxedminipage}
\end{figure}

We begin by fully defining and analyzing Step 3 of $\jnf$.
\begin{lemma}[Rounded Approximate Eigenvalue Powers]\label{lem:approxpowers} The approximate powers of the eigenvalues $\ax{\lambda_{ij}^p}\in \C_\b\bit{b'/b'}$ required in Step 3 may be computed in $\O(n^2b')$ bit operations and satisfy
$$ |\ax{\lambda_{ij}^p} - {\lambda_{ij}}^p|\le 2^{-b'}\cdot n 2^{2n+1}\|A\|^{n^2+n}\le 2^{-b'+an^2+5an}$$
for every $i,j$.
\end{lemma}
\begin{proof}

 Suppose we wish to compute powers $\lambda,\lambda^2,\ldots,\lambda^r$ for some nonzero eigenvalue $\lambda=\lambda_{ij}$ appearing in $W$. Let $\ax{\lambda}$ be the approximate eigenvalue produced in Step 2, satisfying $|\lambda-\ax{\lambda}|\le 2^{-b'}$.	We use the following inductive scheme for $p=2,\ldots,r\le n$:

	$$ \ax{\lambda^p} := \round_{b'}(\ax{\lambda^{p-1}}\cdot \ax{\lambda}).$$

First, observe that from Step 2 we have the error estimate $|\lambda-\ax{\lambda}|\le 2^{-b'}$, which implies that for every $p\le n$:
\begin{equation}\label{eqn:round1}
    |\lambda^p-(\ax{\lambda})^p|\le 2^{-b'}\cdot p\cdot|\lambda^{p-1}|\le 2^{-b'}n\|A\|^n\le 2^{-b'}\cdot n 2^{2n}\|A\|^{n^2+n}
\end{equation} 
since $|\lambda|\le \|A\|$ and $\|A\|\ge 1$. 
Thus, it suffices to show that for each $p$:
\begin{equation}\label{eqn:round2} |\ax{\lambda^p}-(\ax{\lambda})^p|\le 2^{-b'}\cdot n 2^{2n}\|A\|^{n^2+n}.\end{equation}
Notice that $|\lambda|\ge \|A\|^{-n}$ since the product of the nonzero eigenvalues of $A$ is given by $e_k(A)\ge 1$, for $e_k$ the last nonzero elementary symmetric function of $A$, and each eigenvalue of $A$ is at most $\|A\|$. Since 
\begin{equation}
    2^{-b'}\le 2^{-an-2\lg n-1}\le \|A\|^{-n}/2,
\end{equation} we have $|\ax{\lambda}|\ge \|A\|^{-n}/2$ and thereby $|(\ax{\lambda})^p|\ge \|A\|^{-n^2}/2^n$ for every $p=1,\ldots n$. It now follows by induction that:
	$$|\round_{b'}(\ax{\lambda^{p-1}}\cdot \ax{\lambda})-(\ax{\lambda})^p)|\le 2^{-b'}p\cdot 2^{n}\|A\|^{n^2}|(\ax{\lambda})^p|$$
	for every $p=2,\ldots,r$, i.e., where the inductive hypothesis is that in each step the rounding incurs a {\em relative} error of at most $2^{-b'}2^n\|A\|^{n^2}$, and we observe that the relative errors simply add up since they are sufficiently smaller than one. Since we also have the upperbound $|(\ax{\lambda})^p|\le 2^n\|A\|^n$, the desired inequality
	 \eqref{eqn:round2} follows. Combining this with \eqref{eqn:round1} yields the advertised error bound.
	
The total bit complexity for one eigenvalue is $n$ times the cost of one step of the induction, which is $\O(nb')$. Since there are $n$ eigenvalues, the total cost is $\O(n^2b')$. 
\end{proof}

The key condition number bounds used in proving correctness of $\jnf$ are the following, obtained via the minimum eigenvalue gap of $W$ which is controlled using the maximum bit length of the $C_i$. Item (iii) is also used in the analysis of the spectral factorization algorithm in Section 3.
\begin{lemma}[Condition Numbers from Gaps]\label{lem:jnfcond} If $A\in\Z^{n\times n}\bite{a}$ and $A=UFU^{-1}=(UW)J(UW)^{-1}=VJV^{-1}$ for exact Frobenius and Jordan forms as above, then
\begin{enumerate}
\item [(i)] $\kappa(U)\le n^2\cdot n!\cdot 2^{n\size{U}+\size{U}}\le 2^{\O(an^3)}.$
\item [(ii)] $\kappa(W)\le 2^{\O(an^3)}.$
\item [(iii)] $\|V\|\le n2^{\bite{U}}\cdot 2^{n+an+n\lg n}\le 2^{O^*(an^2)}\quad\textrm{and }\|V^{-1}\|\le n\cdot (n!)^22^{n\bite{U}+cn^2+2n^2\lg n}\le  2^{\O(an^3)}.$
\end{enumerate}
\end{lemma}
\begin{proof} For (i), note that $\|U\|\le n2^{\size{U}}$ and since $U\in\Z^{n\times n}\bit{an^2}$, we have $$\|U^{-1}\|\le n\cdot n!2^{n\bite{U}}\le 2^{\O(an^3)}.$$ Consequently $\|U\|\|U^{-1}\|\le 2^{\O(an^3)}$. 

The matrix $W$ is a direct sum of $W_i$, which are  confluent Vandermonde matrices in the eigenvalues $\lambda_{ij}$, which are roots of the $\chi_{C_i}(x)$. By Theorem \ref{thm:mignotte} and  $\size{\chi_{C_i}}=\size{C_i}$, we have
$$ \delta:=\min_i \left[\gap(\chi_{C_i})\right]\ge 2^{-\max_i\size{\chi_{C_i}}n-2n\lg n} = 2^{-cn-2n\lg n}\ge 2^{-\O(an^2)},$$
where the last inequality uses $c=\O(an)$.
 Then  \cite[Theorem 1]{batenkov2012norm} implies that $$\|W^{-1}\|\le n!(1/\delta)^n\le n!2^{cn^2+2n^2\lg n}\le  2^{\O(an^3)}.$$ On the other hand, the formula \eqref{eqn:brand} reveals that $\|W\|\le n\cdot 2^{n+an+n\lg n}$ since $|\lambda_{ij}|\le n2^a$. Multiplying these two bounds yields (ii). 
 
 Finally, we have $\kappa(V)\le \kappa(U)\kappa(W)$ and $\|V^{-1}\|\le \|W^{-1}\|\|U^{-1}\|$, establishing (iii).
\end{proof}

\begin{theorem}\label{thm:jnfmain} The algorithm $\jnf$ satisfies its guarantees and runs in expected $\O(n^{\omega+3}a+n^4a^2+n^\omega b)$ bit operations.
\end{theorem}
\begin{proof}

{\em Bit Length of the Output.} The bit length assertions in Steps 1 and 2 are immediate from Theorem \ref{thm:storjohann} and Corollary \ref{cor:rootmult}. The bit length of $\ax{J},\ax{W}$ in Step 3 is guaranteed by Lemma \ref{lem:approxpowers}. The bit length of the product in Step 4 is implied by Fact \ref{fact:arithmetic}.\\
\noindent {\em Error Bounds.} Step 1 is exact. 

Lemma \ref{lem:approxpowers} implies that the matrices 
 $\ax{J_i},\ax{W_i}$ in Step 3 satisfy
\begin{equation}\label{eqn:japprox1}
    \maxn{J_i-\ax{J_i}}\le 2^{-b'},\qquad \maxn{W_i-\ax{W_i}}\le 2^{-b'+n+an^2+5an}
\end{equation} 
This additive bound is preserved under taking direct sums. To obtain the multiplicative bound, we observe that $\norm{W}\ge 1$ by \eqref{eqn:brand}; if there is a Jordan block of size at least two then $\norm{J}\ge 1$ also, otherwise since $A$ is integral we have 
$$\prod_{\textrm{nonzero} \lambda_{ij}}\lambda_{ij}^{\mathrm{mult}(\lambda_{ij})}=e_k(A)=e_k(J)\ge 1$$
for the last nonzero elementary symmetric function $e_k$ of $A$, so one of the eigenvalues $\lambda_{ij}$ must be at least $\|A\|^{-n}$ and we have crudely $\|J\|\ge \|A\|^{-n}\ge 2^{-2an}$. In either case, we conclude after passing to the operator norm that
$$\norm{J-\ax{J}}\le  2^{-b'+2an}\norm{J}\le 2^{-b}\|J\|\quad\textrm{ and }\quad 
	\norm{W-\ax{W}}\le 2^{-b'+an^2+6an}\norm{W}.$$
	
To obtain the final error bound on $\ax{V}$ in Step 4, we observe that
\begin{equation}\label{eqn:vapprox1}
    \|V-\ax{V}\|=\|UW-U\ax{W}\|\le \|U\|n2^{-b'+n+an^2+5an}\le 2^{-b'+an^2+6an+\size{U}+2\lg n}
\end{equation}
since $\|U\|\le 2^{\size{U}+\lg n}$. Since $\|V\|\ge \|W\|/\|U^{-1}\|\ge 2^{-n\size{U}-n\lg n-\lg n}$, we obtain the conclusion
\begin{equation}\label{eqn:vapprox}
\|V-\ax{V}\|\le 2^{-b'+an^2+6an+\size{U}+3\lg n+n\size{U}+n\log n}\|V\|\le 2^{-b}\|V\|
\end{equation}
by our choice of $b'$, as desired.\\

\noindent {\em Condition of $\ax{V}$.}
	The bound \eqref{eqn:vapprox1} together with Lemma \ref{lem:jnfcond}(iii) implies $$\|V-\ax{V}\|\|V^{-1}\|\le  2^{-b'+an^2+6an+\size{U}+2\lg n} \cdot  n\cdot (n!)^22^{n\bite{U}+cn^2+2n^2\lg n}\le 1/2$$
	by the choice of $b'$ in Step 2. It follows from 
	\eqref{eqn:kappakappa} that
	\begin{equation}\label{eqn:kappavax}\kappa(\ax{V})\le \kappa(V)\frac{1+1/2}{1-1/2}\le 3\kappa(V)\le 2^{\O(an^3)},\end{equation}
		as desired.\\

\noindent {\em Complexity.}
	Step 1 takes $\O(n^5a+n^4a^2)$ bit operations by Theorem \ref{thm:storjohann}.  
	
	Step 2 takes $\O(n^2(an^3+b))$ bit operations by Theorem \ref{thm:pan}.
	
	Step 3 takes $\O(an^5+bn^2)$ bit operations by Lemma \ref{lem:approxpowers}.
	
	The matrix multiplication in Step 4 $\O(n^\omega (b'+an^2))$ time.
	
	The total running time is therefore $\O(n^\omega b'+n^4a^2)=\O(n^{\omega+3}a+n^{4}a^2+n^{\omega}b)$, as advertised.
	
\end{proof}

\begin{corollary}[JNF of Rational Matrices with Common Denominator]\label{rem:commonjnf} The algorithm $\jnf$ can easily be used to compute the JNF of $A/q$ for integer $A$ and $q$: if $\ax{J},\ax{V}$ is an approximate JNF of $A$ with $b$ bits of accuracy, then $\ax{J}/q,\ax{V}$ is an approximate JNF of $A/q$ with $b-\lg(q)$ bits of accuracy. This fact will be useful in the spectral factorization algorithm in the following section.

\begin{remark}
The proof of Theorem \ref{thm:jnfmain} yields an explicit estimate on $\kappa(V)$ in terms of the bit lengths $a,c, \size{U}$ which may be better than the worst case bound of $2^{\O(an^3)}$ for specific instances.
\end{remark}
\end{corollary}

\section{Spectral Factorization}\label{sec:specfact}
We briefly review some aspects of the theory of matrix polynomials (the reader may consult \cite{gohberg2005matrix} for a comprehensive introduction). Given a
monic matrix polynomial $L(x)=x^dI+\sum_{i\le d-1}x^iL_i$ with
$L_i\in\C^{n\times n}$, its adjoint is $L^*(x)=x^dI+\sum_{i\le d-1}x^iL_i^*$,
and its latent roots are the $\lambda\in\C$ such that $L(\lambda)$ is singular. 
The {\em block companion matrix}\footnote{This is a ``row'' companion matrix as opposed to the ``column'' companion matrices
in Section \ref{sec:jnf}. This is customary in the theory of matrix polynomials.} of $L$ is the $dn\times dn$ matrix:
\begin{equation}\label{eqn:blockcomp}
	C_L:=\bm{ 	0& 1	& 	& 	&	&	\\
			0&	0&	1&	&	&	\\
			&	&	\vdots 	&	&	\\
			&	&	&	&	&1	\\
			-L_0&	-L_1&	-L_2&	-L_3& \ldots& -L_{d-1}}
\end{equation} The following important theorem states the existence of spectral factorizations of positive definite monic matrix polynomials,
and gives a way of computing them using the block companion matrix.

\begin{theorem}[Theorems 5.1, 5.4 of \cite{gohberg1980spectral}] \label{thm:sfglr1} Suppose $P(x)=P^*(x)=x^{2d}I+\sum_{i\le 2d-1} x^{i}P_i\in\C^{n\times n}[x]$ monic of degree $2d$ satisfies $P(x)\succeq 0$ for all $x\in\R$. Then:
	\begin{enumerate}
	\item There is a unique monic $Q(x)\in\C^{n\times n}[x]$ of degree $d$ such that $P(x)=Q^*(x)Q(x)$ and $Q$ has all of its latent roots in the closed upper half plane. 
	\item The complex eigenvalues of $C_P$ occur in conjugate pairs, and each Jordan block in the JNF of $C_P$ corresponding to a real eigenvalue has even size. 

	\item Let $C_P=VJV^{-1}$ be a Jordan Form of the block companion matrix of $P$ with block decomposition
		$$ J=:\bm{ J_+ & &\\ &   J_0 & \\ & & J_-},\qquad V=:\bm{V_+ & V_0 & V_-\\ Z_+ & Z_0 & Z_-}$$
			for $J_{\pm}$ corresponding to eigenvalues in the open upper/lower half plane and $J_0$ corresponding to the real eigenvalues and $V_\pm, V_0$ having $dn$ rows.
			Then 
			\begin{equation}\label{eqn:cqformula}C_Q = V_{\ge 0} J_{\ge 0} V_{\ge 0}^{-1},\end{equation}
			where 
			\begin{equation}\label{def:vgeq} V_{\ge 0}=[V_+, V_0^\half]\in\C^{dn\times dn}\quad\textrm{and}\quad J_{\ge 0}=J_+\oplus J_0^\half\in\C^{dn\times dn}.\end{equation}
			Here, for each Jordan block of size $2s$ in $J_{0}$,  $J_0^\half$ contains a Jordan block of size $s$ with the same eigenvalue, and $V_0^\half$  contains as columns the first $s$ of the corresponding $2s$ columns of $V_0$.
	\end{enumerate}
\end{theorem}
The formula \eqref{eqn:cqformula} gives a one line algorithm for computing $Q$ given access to the exact Jordan form of $P$. The key issue
is that in order to use an approximate Jordan form $\ax{V}\ax{J}\ax{V^{-1}}$ in the formula, we must have a good bound on the condition number of $V_{\ge 0}$ 
in order to control the error incurred during inversion. Note that while Lemma \ref{lem:jnfcond} guarantees a bound on $\kappa(V)$,
this does not in general imply a bound on its submatrices; indeed, it is known that there can be square submatrices of $V$ which are singular.
The main technical contribution of this section is to prove a bound on $\kappa(V_{\ge 0})$ in terms of $\kappa(V)$ by exploiting the special structure of $V$
which arises from the structure of $C_P$. This is encapsulated in the following fact, which may be found in any reference on matrix polynomials (e.g., \cite[\S 1]{gohberg2005matrix}).

\begin{fact} If $C_P=VJV^{-1}$ is the Jordan normal form of an $n\times n$ complex matrix polynomial $P$ of degree $d$, then there is a matrix   $X\in\C^{n\times 2dn}$ such that:
			\begin{equation}\label{eqn:xjstack}V=\bm{ X\\ XJ\\ XJ^2\\ \vdots \\ XJ^{2d-1}}.\end{equation}
\end{fact}
We show that the least singular value of a column submatrix of any matrix of type \eqref{eqn:xjstack} may be related to the least singular values of certain block submatrices.
\begin{lemma}[Condition of Submatrices of Companion JNF]\label{lem:sfcond} Given any $Y\in\C^{n\times D}$ and $K\in \C^{D\times D}$ with $\|K\|\ge 1$, define for $k=1,2,\ldots$ the $nk\times D$ matrices:
$$ W_k:=\bm{ Y\\ YK\\ YK^2\\ \vdots \\ YK^{k-1}}.$$
Then
$$\sigma_D(W_D)\ge \frac{\sigma_D(W_{k})}{\sqrt{k}(4\|K\|)^{D(k-D+1)}}$$
for every $k\ge D$.
\end{lemma}
\begin{proof} Suppose $x\in\C^{D}$ is a unit vector satisfying $\|W_Dx\|=\sigma_D(W_D)=:\sigma$. We will show that 
\begin{equation}\label{eqn:yksigma}\|W_kx\|\le \sigma\cdot \sqrt{k} (2D^{1/D}\|K\|)^{D(k-D+1)},\end{equation}
which yields the Lemma by using $D^{1/D}\le 2$.
Let $q$ be the characteristic polynomial of $K$. By the Cayley-Hamilton theorem, we have
$$q(K) = K^D + \sum_{0\le i\le d-1} c_i K^{i} = 0,$$
for some complex coefficients $c_i$ crudely bounded as 
$$\max_{i\le D-1} |c_i| \le 2^D\|K\|^D:=\alpha,$$
by considering their expansion as elementary symmetric functions in the eigenvalues of $K$. Using this expression, we obtain the identity:
\begin{align*}
YK^{j}x = YK^{j-D}K^Dx = -\sum_{0\le i\le D-1} c_i YK^{j-D}K^ix,
\end{align*}
for every $j\ge D$.
By the triangle inequality, this yields:
\begin{align*}
    \|YK^j x\|\le \alpha D\cdot \max_{i<j}\|YK^i x\|,
\end{align*}
which applied recursively gives: $$\|YK^jx\|\le (\alpha D)^{j-D+1}\cdot \max_{i<D}\|YK^ix\|\le (\alpha D)^{j-D+1}\sigma.$$
Summing over all $j\le k$, we have:
$$\|W_kx\|^2 \le \sigma^2 + \sum_{j=D}^k (\alpha D)^{2(j-D+1)}\sigma^2 \le k(\alpha D)^{2(k-D+1)}.$$

Taking a square root establishes \eqref{eqn:yksigma} and finishes the proof.
\end{proof}
\begin{remark} Lemma \ref{lem:sfcond} is a quantitative version of the main claim of \cite[\S 2.3]{gohberg1980spectral} showing that $V_{\ge 0}$ is invertible whenever $V$ is invertible, which is central to the theory of matrix polynomials.
The proof above is an arguably simpler proof of this fact, and may be of independent interest. The original proof of \cite{gohberg1980spectral} relies on a delicate analysis of a certain indefinite quadratic form.\end{remark}

Finally, we are able to bound $\kappa(V_{\ge 0})$.
\begin{lemma}\label{lem:kappavge} In the setting of Theorem \ref{thm:sfglr1},
$$ \|V_{\ge 0}^{-1}\|\le \|V^{-1}\|\cdot \sqrt{2dn} (4+4\|C_P\|)^{dn(dn+1)}$$
and
$$\kappa(V_{\ge 0})\le \kappa(V)\cdot \sqrt{2dn} (4+4\|C_P\|)^{dn(dn+1)}.$$
\end{lemma}
\begin{proof}
Letting $C_P=VJV^{-1}$, the similarity $V$ has the form \eqref{eqn:xjstack} for some $X\in\C^{n\times 2dn}$. Let $X_{\ge 0}$ be the $n\times dn$ submatrix of $X$ with columns corresponding to the columns in $V_{\ge 0}$. Apply Lemma \ref{lem:sfcond} with $D=dn, k=2dn, K = J_{\ge 0}, Y = X_{\ge 0}$, noting that $\|K\|\le 1+\|C_P\|$ since all of the diagonal entries of $J_{\ge 0}$ are eigenvalues of $C_P$ and bounded by its norm. This yields:
$$ \sigma_{dn} (V_{\ge 0}) \ge \frac{\sigma_{dn}\left(\bm{ V_{\ge 0} \\ Z_{\ge 0}}\right)}{\sqrt{2dn} (4+4\|C_P\|)^{dn(dn+1)}},$$
where $Z_{\ge 0}$ has the obvious meaning, yielding the first claim. But $\bm{ V_{\ge 0} \\ Z_{\ge 0}}$ is a column submatrix of $V$ so $$\sigma_{dn}\left(\bm{ V_{\ge 0} \\ Z_{\ge 0}}\right)\ge \sigma_{dn}(V)\ge \sigma_{2dn}(V).$$ Combining this with $\sigma_1(V_{\ge 0})\le \sigma_1(V)$, we obtain the second claim.
\end{proof}

We now present the algorithm $\specfact$ which approximately computes the $Q(\cdot)$ guaranteed by Theorem \ref{thm:sfglr1} using an approximate Jordan normal form computation and exact inversion. We rely on the following tool from symbolic computation.

\begin{theorem}[Fast Exact Inversion, \cite{storjohann2015complexity}]\label{thm:storinv} There is a randomized algorithm which given an invertible matrix $A\in\Z^{n\times n}\bite{a}$ exactly computes its inverse $A^{-1}\in\Q^{n\times n}\bit{an/an}$ in time $\O(n^3a+n^3\log\kappa(A))$. 
\end{theorem}

\begin{figure}[ht]
\begin{boxedminipage}{\textwidth}
{\bf Algorithm $\specfact$}: \quad Input: Coefficients $P_0,\ldots,P_{2d-1}\in\Q^{n\times n}\bite{a/a}$ (with a common denominator) of a monic matrix polynomial $P(x)$, desired bits of accuracy $b\in\N$.\\
Output: $\ax{Q_0},\ldots,\ax{Q_{d-1}}\in\C_\b^{n\times n}\bit{a(dn)^3+b}$ or
or a certificate that $P(x)\nsucceq 0$ for some $x\in\R$.\\
Guarantee: If $P(x)\succeq 0$ then  $\maxn{\ax{Q}(\cdot)-Q(\cdot)}\le 2^{-b}\maxn{Q(\cdot)}$ for $P(x)=Q^*(x)Q(X)$.
\begin{enumerate}
	\item Compute an approximate Jordan Normal Form $(\ax{V},\ax{J})=\jnf(C_P,b'')$ of $C_P$ using Corollary \ref{rem:commonjnf}, with $\ax{J},\ax{V}\in \C_\b^{n\times n}\bit{b''/b''}$ where $b''$ is chosen to be the least integer such that $$f_1(b'')+f_2(b'')\le 2^{-b}\|P(\cdot)\|_{max},$$
	where $f_1,f_2$ are defined in \eqref{eqn:lasttwo},\eqref{eqn:firstterm}.
	Determine its eigenvalues on\footnote{This can be determined by testing if there are any approximate eigenvalues $\ax{\lambda}$ with $|\ax{\lambda}-\overline{\ax{\lambda}}|\ll 2^{-an^2d^2}$ by Theorem \ref{thm:mignotte}.}, below, and above the real line. If any Jordan block corresponding to a real eigenvalue has odd size, output ``$P(x)\nsucceq 0$''.  
    \item Let
		$$ \ax{J}=:\bm{ \ax{J_+} & &\\ &   \ax{J_0} & \\ & & \ax{J_-}},\qquad V=:\bm{\ax{V_+} & \ax{V_0} & \ax{V_-}\\ * & * & *}$$
		be a block decomposition such that $\ax{J_+}$ corresponds to eigenvalues of $C_P$ in the open upper half plane and $\ax{J_0}$ corresponds to real eigenvalues of $C_P$.
	\item Output the negative of the last row of 		\begin{equation}\ax{C_Q} := \ax{V_{\ge 0}} \ax{J_{\ge 0}} \ax{V_{\ge 0}^{-1}},\end{equation}
			where 
			\begin{equation} \ax{V_{\ge 0}}:=[\ax{V_+}, \ax{V_0}^\half]\in\C^{dn\times dn}\quad\textrm{and}\quad \ax{J_{\ge 0}}:=\ax{J_+}\oplus \ax{J_0}^\half\in\C^{dn\times dn}\end{equation}
		 and  $(\cdot)^\half$ is defined as in Theorem \ref{thm:sfglr1}. The approximate inverse $\ax{V_{\ge 0}^{-1}}$ is computed by exactly computing $(\ax{V_{\ge 0}})^{-1}$ using Theorem \ref{thm:storinv} and letting 
		 $\ax{V_{\ge 0}^{-1}}=\round_{b''}((\ax{V_{\ge 0}})^{-1})$. 
\end{enumerate}
\end{boxedminipage}
\end{figure}

\begin{theorem}\label{thm:specfactmain} The algorithm $\specfact$ satisfies its guarantees and runs in $\O((dn)^6a+(dn)^4a^4+(dn)^3b)$ bit operations.
\end{theorem}
\begin{proof}
Item (2) of Theorem \ref{thm:sfglr1} shows that $P(x)\nsucceq 0$ if there is an odd size Jordan block with real eigenvalue. 

Assuming this is not the case, that theorem shows that the exact spectral factor $Q$ is given by the last row of $C_Q=V_{\ge 0}J_{\ge 0}V_{\ge 0}^{-1}$. We now prove that the quantity $\ax{V_{\ge 0}}\ax{J_{\ge 0}}\ax{J_{\ge 0}}^{-1}$ computed by $\specfact$ is close to $C_Q$. This is a consequence of the following estimates. Given Lemma \ref{lem:kappavge}, the  arguments are essentially identical to those in the proof of Lemma \ref{lem:jnfcond} and Theorem \ref{thm:jnfmain} (the key point being that the inverse of a well-conditioned matrix is stable under small enough perturbations).
\newcommand{\jge}{J_{\ge 0}}
\newcommand{\vge}{V_{\ge 0}}

Let $B$ be the maximum of $n2^{a}$ (which is an upperbound on $\|C_P\|$) and the two explicit upper bounds on $\|V\|$ and $\|V^{-1}\|$ in Lemma \ref{lem:jnfcond}(iii) (noting that the bit size $\bite{U}$ can be read off from the matrix $U$ produced during the execution of $\jnf$, so $B$ is easily computable).  It follows by Lemma \ref{lem:kappavge} and the guarantees of JNF that:
\begin{equation}
    \|V\|,\|J\|,\|\vge\|,\|\jge\|\le B,\qquad \|\ax{\vge}\|,\|\ax{\jge}\|\le 2B,\qquad \|\vge^{-1}\|\le 2^{2(a+\lg n)d^2n^2}B=:B'.
\end{equation}

Applying the triangle inequality thrice, we decompose the output error of $\specfact$ as:
\begin{align*}
    \|\ax{\vge}\ax{\jge}\ax{\vge^{-1}}-\vge \jge \vge^{-1}\| &\le \|\ax{\vge}\ax{\jge}\ax{\vge^{-1}}-\ax{\vge}\ax{\jge} \vge^{-1}\|\\
    &+ 
    \|\ax{\vge}\ax{\jge}\vge^{-1}-\ax{\vge}\ax{\jge} \vge^{-1}\|\\
        &+ 
    \|\ax{\vge}\jge\vge^{-1}-\vge\jge \vge^{-1}\|\\
    &\le 4B^2\|\ax{\vge^{-1}}-\vge^{-1}\| \\&+ B'\cdot 2B^2\|\ax{\jge}-\jge\| \\&+B'\cdot B^2 \|\ax{\vge}-\vge\|.
\end{align*}

The sum of the last two terms is  bounded by
\begin{equation}\label{eqn:lasttwo}2^{2(a+\lg n)d^2n^2}\cdot 2B^2 (2^{-b''}\|J\|+2^{-b''}\|V\|) \le 2^{-b''+2(a+\lg n)d^2n^2}\cdot 2B^3 =:f_1(b'').\end{equation}

For the first term, observe that whenever 
\begin{equation}\label{eqn:bsmall}
    2^{-b''}\le B'/2
\end{equation}
we have
\begin{align}4B^2\|\ax{\vge^{-1}}-\vge^{-1}\|&=4B^2\|\round_{b''}((\ax{V_{\ge 0}})^{-1})-\vge^{-1}\|
\\&\le 4B^2(\|\round_{b''}((\ax{V_{\ge 0}})^{-1})-(\ax{V_{\ge 0}})^{-1}\| + \| (\ax{V_{\ge 0}})^{-1} - \vge^{-1}\|)
\\&\le 4B^2(n2^{-b''}+\frac{2^{-b''}\|\vge^{-1}\|}{1-2^{-b''}\|\vge^{-1}\|}\|\vge^{-1}\|
\\&\le 4B^2(n2^{-b''}+\frac{2^{-b''}B'}{(1/2)}B')
\\\label{eqn:firstterm}&\le 4B^2(2n2^{-b''}(B')^2) =: f_2(b'').
\end{align}
By our choice of $b''$ in Line 1, the advertised error bound for the output
	follows. Note that $B=2^{O^*(a(dn)^3)}$ in the worst case.

{\em Complexity.} The running time of $\jnf$ in Step 1 is $\O((dn)^{\omega+3}a+(dn)^4a^2+(dn)^\omega b'')$. Step 2 does not involve any computation. The time taken to exactly invert $\ax{\vge}$ in Step 3 using Theorem \ref{thm:storinv} (after pulling out the common dyadic denominator to obtain an integer matrix) is $\O((dn)^3\cdot (b''+a(dn)^3+b))$ by the estimate $\kappa(\ax{\vge})=2^{O^*(a(dn)^3)}$ which follows from \eqref{eqn:kappakappa} and the bound on the first term above. The time taken to round down the entries of $\ax{\vge}^{-1}$ to $b''$ bits is $\O((dn)^2b'')$. The time taken to multiply together the three matrices is $\O((dn)^\omega b'')$. Thus, the total number of bit operations is dominated by $$\O((dn)^{\omega+3}a+(dn)^4a^2+(dn)^3 b'') = \O((dn)^6a+(dn)^4a^2+(dn)^3b),$$ as advertised.
\end{proof}

\begin{corollary}[Spectral Factorization of Non-Monic Polynomials]\label{cor:nonmon} Suppose $P(x)=VV^*x^{2d}+\sum_{i=0}^{2d-1} x^iP_i$ is a positive semidefinite Hermitian matrix polynomial with $V, P_i\in\Q^{n\times n}\bit{a/a}$ with a common denominator and $V$ invertible. Then an approximate spectral factorization of $P(x)$ accurate to $b-\O(a)$ bits (as in Theorem \ref{thm:specfactmain}) can be computed in expected $\O((dn)^6an+(dn)^4(an)^2+(dn)^3b)$ bit operations.
\end{corollary}
\begin{proof}
The rescaled polynomial $\tilde{P}(x):= x^{2d}I + \sum_{i=0}^{2d-1} x^i V^{-1}P_i V^{-*}$ is also positive semidefinite,  has coefficients in $\Q^{n\times n}\bit{an/an}$, and is monic. Applying Theorem \ref{thm:specfactmain} yields an approximate spectral factor $\tilde{Q}$, for which $Q(x) = V \tilde{Q}(x) V^*$ is an approximate spectral factor of $P(x)$ with at most a loss of $\O(a)$ bits of accuracy.
\end{proof}

\section{Discussion and Future Work}
\label{sec:discussion}

For historical context, proving bit complexity bounds on fundamental linear algebra computations
(such as inversion, polynomial matrix inversion, Hermite/Smith/Frobenius normal forms \cite{kannan1979polynomial,kannan1985solving,storjohann1997fast,storjohann1998n,storjohann2001deterministic, gupta2011computing,zhou2015deterministic,kaltofen2015complexity}) has been
a vibrant topic in theoretical computer science and symbolic computation since the 70's, with near-optimal arithmetic and bit complexity
bounds being obtained for several of these problems within the last decade (e.g. \cite{storjohann2015complexity}). However,
this program did not reach the same level of completion for problems of a spectral nature,
such as the ones studied in this paper. While the polynomial time bounds obtained in this paper are modest,
we hope they will stimulate further work on these fundamental problems, as well as the important special case of efficiently diagonalizing a diagonalizable matrix in the forward error model, which remains unresolved (in that we don't know the correct exponent of $n$; the recent work \cite{banks2020pseudospectral} obtains nearly matrix multiplication time for the backward error formulation of the problem).

Some concrete questions left open by this work are:
\begin{enumerate}
    \item Improve the running time for computing the JNF of a general matrix. The best known running time for computing the eigenvalues of a matrix is roughly $O(n^{\omega+1}a)$ \cite{pan1999complexity}, so this seems like a reasonable goal to shoot for. The current bottleneck is the bound of $2^{\O(an^3)}$ on the condtion number of the similarity $V$, which could conceivably be improved to $2^{\O(an^2)}$.
    \item Improve the running time for computing the JNF of the block companion matrix of a matrix polynomial by exploiting its special structure, particularly \eqref{eqn:xjstack}. This would yield faster algorithms for spectral factorization.
\end{enumerate}

\subsubsection*{Acknowledgments}
We thank the anonymous referees of a previous version of this paper, whose thoughtful comments greatly improved the presentation. We thank Bill Helton, Cl\'ement Pernet, Pablo Parrilo, Mario Kummer, Rafael Oliveira, and Rainer Sinn for helpful discussions, as well as the Simons Institute for the Theory of Computing, where a large part of this work was carried out during the ``Geometry of Polynomials'' program.
\bibliographystyle{alpha}
\newcommand{\etalchar}[1]{$^{#1}$}

\end{document}